\documentclass[conference]{IEEEtran}

\usepackage{amsmath}
\usepackage{amsfonts}
\usepackage{mathtools}
\usepackage{amsthm}
\usepackage{graphicx}
\usepackage[T1]{fontenc}
\usepackage{color}

\newtheorem{theorem}{Theorem}

\title{XOR-based Source Routing}

\author{\IEEEauthorblockN{J\'er\^ome Lacan and Emmanuel Lochin}
\IEEEauthorblockA{ISAE-SUPAERO\\
Universit\'e de Toulouse\\
Email: \{jerome.lacan,emmanuel.lochin\}@isae-supaero.fr}
}

\date{\today}

\begin{document}
\maketitle

\begin{abstract}
We introduce a XOR-based source routing (XSR) scheme as a novel approach to enable fast forwarding and low-latency communications. XSR uses linear encoding operation to both 1)~build the path labels of unicast and multicast data transfers; 2)~perform fast computational efficient routing decisions compared to standard table lookup procedure without any packet modification all along the path. XSR specifically focuses on decreasing the complexity of forwarding router operations. This allows packet switches (e.g, link-layer switch or router) to perform only simple linear operations over a binary vector label which embeds the path. XSR provides the building blocks to speed up the forwarding plane and can be applied to different data planes such as MPLS or IPv6. Compared to recent approaches based on modular arithmetic, XSR computes the smallest label possible and presents strong scalable properties allowing to be deployed over any kind of core vendor or datacenter networks. At last but not least, the same computed label can be used interchangeably to cross the path forward or reverse in the context of unicast communication.
\end{abstract}

\section{Introduction}
\label{sec:introduction}
Source routing is a very old technique to route a data packet from a source to a destination, initially presented in \cite{rfc791} and currently developed at the IETF within the SPRING (Source Packet Routing in Networking) working group \cite{rfc7855}. Compared to conventional routing that forwards packets following both the IP destination address and the forwarding table lookup, source routing allows the sender to partly or completely indicate inside the packet headers the path that must be followed. Source routing brings out several advantages. As highlighted in \cite{2015:Lee}, the data plane becomes simpler, core elements perform simple operations and traffic engineering is more flexible.

Source routing technique gained in popularity in particular following the rapid spread of Software Defined Networking (SDN) paradigm as a scalable solution to deploy services in datacenters \cite{2016:Abujoda}. 
In particular, the authors in \cite{2018:SDN-SR} illustrate that SDN-based source routing significantly decrease flow-states exchange by storing the path information into packet headers. Encoding the whole path inside a packet suppresses expensive lookup procedures inside core packet switches (e.g, link-layer switch or router) as each switch is able to quickly identify the next hop of the path stored in the packet. 

The length of the encoded path label is one of the potential issues in source routing. In particular, there are use cases where each individual hop must be specified in the label resulting in a long list of hops that is instantiated into a MPLS label stack (in the MPLS data plane) or list of IPv6 addresses (in the IPv6 data plane) \cite{rfc7855}. Obviously, this leads to potentially oversized labels. Furthermore, current MPLS equipments only support limited number of stacked labels (five to ten labels are currently supported by some routers \cite{2017:guedrez}). To cope with this problem, there exists an up to date variant called segment routing \cite{rfc8402} that leverages source routing principle. Segment Routing encodes a path label as a stack composed by node segments (a router) and adjacency segments (a router interface output) \cite{rfc8402} which prevents to record all nodes addresses. 

XOR-based source routing (XSR) scheme is a novel approach to improve data plane operations enabling fast forwarding. The originality of XSR is to conjointly optimize the size of the path label with low switching processing cost while enabling multicast and unicast forwarding. This is explained by the use of linear operations over binary vectors. A large survey browsing previous attempts is proposed in \cite{segmentRoutingBiblio}; eight papers are identified therein. In this study, we choose to select and focus on two recent competitive works that provide path optimization techniques to minimize the size of a path label encoded inside packets. We will mainly discuss XSR against these two solutions proposed respectively in 2017 \cite{2017:optimalPathEncoding} and in 2018 \cite{rdna2018}.
In the latter, the authors propose a whole architecture that lays on modular arithmetic to compute a label number to identify (following a reverse operation) the output switch port considering a unique router ID \cite{rdna2018}. After presenting our proposal, we will show in \ref{sec:comparisons} that RDNA requires larger path label length and performs less computational efficient operations than XSR (i.e. XOR versus modulo operation), in particular for multicast. In \cite{2017:optimalPathEncoding}, the authors propose an elegant algorithm that minimizes the maximum length of any encoded path in the network. The main drawback is the restriction of this solution to unicast exchanges. On the contrary, XSR copes with all these issues enabling unicast and multicast communications at the same processing cost, performing fast and low latency routing decisions without any packet modification.

\section{The Path Encoding Problem}
\label{sec:pathproblem}

Actually, a router only needs to assess the output link(s) corresponding to a given input packet. So, the path of a packet can be encoded by the output links sequence of the routers composing the path. Since the labels of the output links (denoted interface labels in the following) are local to a node, they can be represented by short bit vectors. For example, a node having $3$ output paths can number them $0,\ 1$ and $2$ and thus, uses $2$-bit vectors $(00),\ (01)$ and $(10)$ as interface labels to identify them. This principle, adopted by the authors in \cite{Soliman:2012}, uses short fixed-length interface labels. Note that the number of bits needed to identify each interface label of a router depends on the number of output links. The authors in \cite{2017:optimalPathEncoding} further investigate this approach by using variable-length prefix-codes usually used in lossless compression systems to represent the interface labels of the output links. They show that they can reduce the lengths of the largest encoded paths. 

 With segment routing, all the labels have the same length and each router  considers the first label at the top of the stack in the received packet, processes the packet, then removes this label. The next router uses the next label until the final receiver. When short interface labels are used, this strategy can not be applied because the interface labels size is not necessarily a multiple of $8$ bits. In \cite{Soliman:2012}, each interface label has a fixed length and a hop counter is added to the header allowing to identify the current path position. This counter is then decremented by each router before forwarding the packet involving data modifications. Similarly to \cite{Soliman:2012}, due to variable interface label sizes, a pointer is also needed in \cite{2017:optimalPathEncoding} to point the current position in the encoded path. After reading its corresponding interface label, the router slides the pointer and forwards the packet. 

The localization strategies of the labels in the encoded paths have several implications. First, removing or modifying some parts of the label involves header supplementary data operations and computations. The second consequence is that these strategies are only usable for unicast transmissions. Indeed, if we consider multicast or multipath transmissions, some router must send some packets on several interfaces. However, the header modifications done by the router are only based on its local information and thus it is not possible to make different modifications on the packets sent to the different interfaces. 

Other strategies have been proposed to enable multipath or multicast. In \cite{2004:HB-DSR}, the source builds a Bloom filter which is based on the addresses of the nodes of the path and which is stored in the packet header. At the reception of a packet, each router checks whether the addresses of its neighbours are verified by the Bloom filter and forwards the packet to the valid ones. Since Bloom filters are probabilistic tools, the main difficulty with this scheme is to choose the right parameters of the filter to minimize the ratio of false positive while maintaining a reasonable size.    
Finally for multicast transmissions, where a data can be sent to different interfaces, the interface label can be chosen as a bitmap. For example, the label of a packet that must be sent on the interfaces $0, 4$ and $5$ of a router having $6$ interfaces is $(110001)$. A simple method to generate the encoded path is to concatenate the interface labels. More elaborated strategies presented in \cite{coxCast2014} and \cite{rdna2018} encode the path into an integer number. The routers recover their information by computing the residue of this integer modulo a prime number.   

\section{XOR-based Source Routing in a Nutshell}
\label{sec:XSRnutshell}

Let us start with an illustration of XSR interface labels principle. We recall that an interface label corresponds to an interface IDs of a router or a set of interfaces in case of multicast.
Fig.~\ref{fig:localLabelUnicast} shows an example where an input packet of a unicast transmission coming from the interface $(4=100_b)$ is forwarded to the interface $(3=011_b)$. The interface label of this packet for this router is defined as the XOR between the input interface ID and the output interface ID, \textit{i.e.} $100_b \oplus 011_b = 111_b$. It is obvious that the  output interface ID can be computed from the interface label and the input interface ID ($100_b \oplus 111_b = 011_b$). Similarly, for the reverse path, the input interface ID can be retrieved from the output interface ID and the interface label ($011_b \oplus 111_b = 100_b $). 

\begin{figure}[!htb]
    \centering

    \begin{minipage}{0.25\textwidth}
        \centering
		\includegraphics{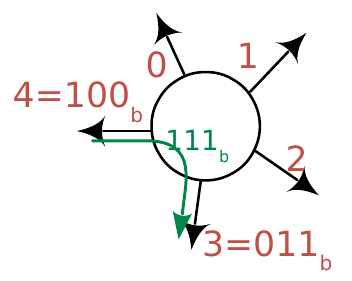}
    	\caption{Unicast interface label}
	    \label{fig:localLabelUnicast}
    \end{minipage}
    \begin{minipage}{0.23\textwidth}
        \centering
		\includegraphics{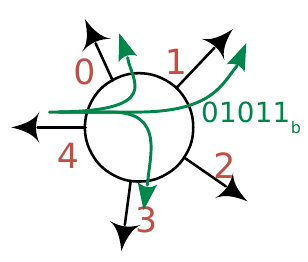}
	\caption{Multicast interface label}
	\label{fig:localLabelMulticast}    
	\end{minipage}
\end{figure}

For a multicast transmission Fig.~\ref{fig:localLabelMulticast}, the interface label to the packet is the bitmap $(01011)$ of the output interface IDs (to be read from the right to the left). A $1$ in positions $0$, $1$ and $3$ means the packet must be forwarded to the ports $0$, $1$ and $3$.  

XSR principle is to concatenate the interface labels of each router of the path into a vector $L$. We assume that each router has a unique identifier $R_{ID}$ (e.g. hardware address). The path label, denoted $P$, is computed by the source (or directly provided by a centralized SDN controller) by applying to $L$ a linear transformation based on the IDs of the routers of the path. This path label is stored in the header of the packet. To forward a packet, each router applies a filtering function (based on its own ID) to the path label to get its interface label. This function is a linear function over the binary finite field $\mathbb{F}_2$ only using XOR-based operations.

The first advantage is that packets are not modified when crossing a router. 
On the contrary to \cite{2017:optimalPathEncoding}, the interface labels list does not need to be ordered in the path label preventing the use of pointer or vector. This filtering function is simply few dot products of short vectors that can be done on-the-fly, compliant with fast routing strategies like e.g. cut-through.     
 
In brief, the length of the path label $P$ is the sum of the lengths of the interface labels of each router of the path, even if they do not have the same length.

\begin{figure}[htb]
	\begin{center}
		\includegraphics[width=0.4\textwidth]{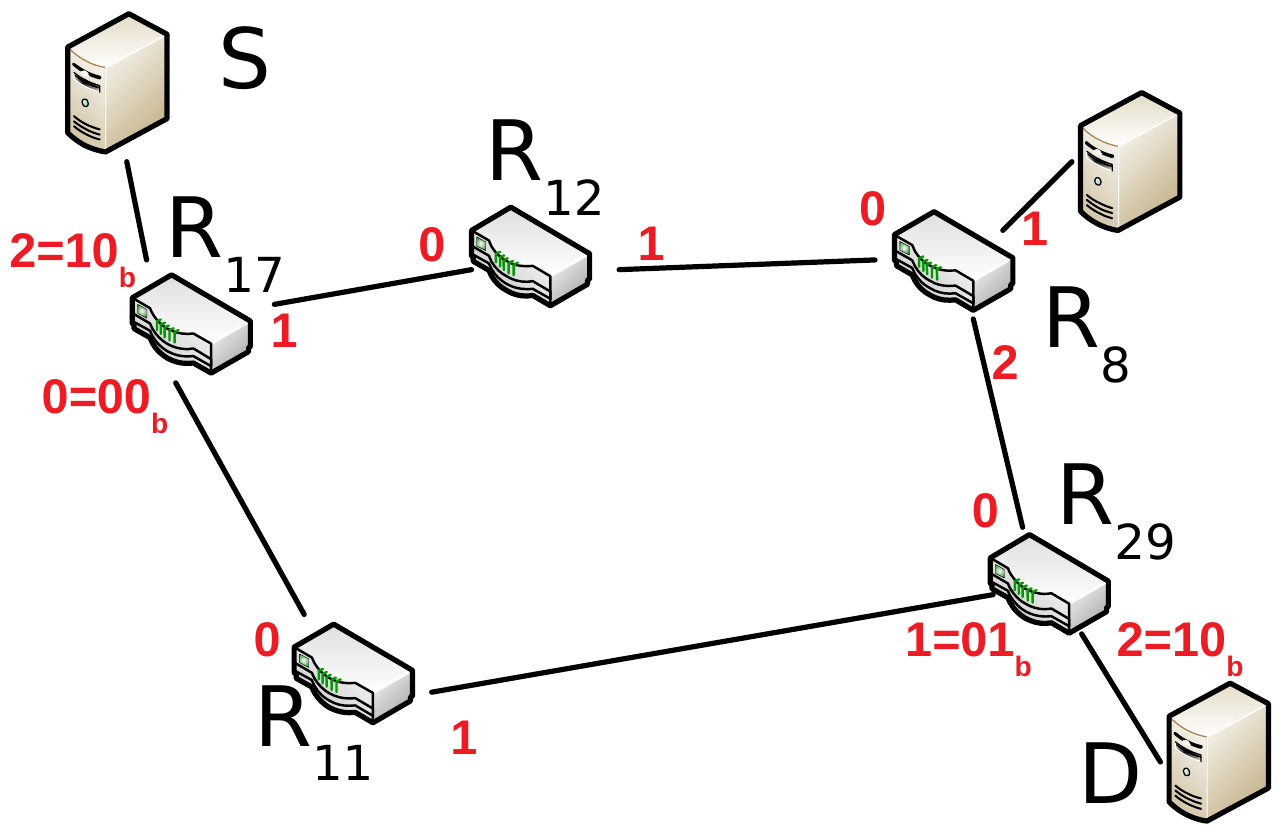}
	\end{center}
	\caption{Example network}
	\label{fig:ExNetwork}
\end{figure}

To illustrate this, let us consider the network presented in Fig.~\ref{fig:ExNetwork}. Assume that the source $S$ requests to send a packet to the destination $D$ through the path $(R_{17}|R_{11}|R_{29})$. The path can be represented by the sequence of interface labels $L=(L_{17}\ L_{11}\ L_{29}) = (2\oplus 0\ 0\oplus 1\ 1\oplus 2)=(10_b\oplus 00_b\ 0_b\oplus 1_b\ 01_b\oplus 10_b)=(10_b\ 1_b\ 11_b)$. Once again, the lengths of the interface labels varies according to the routers or can be variable for the same router like in \cite{2017:optimalPathEncoding}. Here, we consider that $R_{17}$ and $R_{29}$ labels have a length of $2$ bits while $R_{11}$ label has a length of $1$ bit.  

The path label $P$ is computed by solving the following system built from the filtering functions $F_{17},\ F_{11}$ and $F_{29}$:

\begin{eqnarray}
\label{eq:filtering1}
    F_{17}(P)& = &(10)\\ 
\label{eq:filtering2}
    F_{11}(P)&=&(1)\\ 
\label{eq:filtering3}
    F_{29}(P)&=&(11)
\end{eqnarray}

These functions are linear operations characterized by a matrix defined from the routers IDs. For $R_{17}$, let us denote this matrix $M_{17}$. We then have:

$$F_{17}(P) = P\cdot M_{17} =(10) $$

where the notation $\cdot$ between two vectors or matrices represents the matrix multiplication. Since the length of $P$ is the sum of the lengths of the interface labels, i.e. $5$, and the length of the interface label is $2$, $M_{17}$ has $5$ rows and $2$ columns. In this simple example, the first column is defined by the router ID $17=010001_b$ and the second column as its cyclic shift. Finally the label must verify:
\begin{equation}
    F_{17}(P)=P\cdot M_{17} = P.
\begingroup 
\setlength\arraycolsep{2pt}
    \begin{psmallmatrix}
        1 & 1\\ 0 & 1\\ 0 & 0\\ 0 & 0\\ 1 & 0 
     \end{psmallmatrix}
\endgroup
=(10)    
\end{equation}

 By using the same method, we can obtain the following linear equations for $R_{11}$ and $R_{29}$: 
\begin{equation}
    P\cdot M_{11} = P\cdot 
    \begingroup 
    \setlength\arraycolsep{2pt}
    \begin{psmallmatrix}
         0\\ 1\\ 0\\ 1\\ 1 
    \end{psmallmatrix}
    \endgroup
=(1) \textrm{ and }     P\cdot M_{29} = P\cdot 
    \begingroup 
    \setlength\arraycolsep{2pt}
    \begin{psmallmatrix}
       1 & 1\\ 1 & 1\\ 1 & 1\\ 0 & 1\\ 1 & 0  
    \end{psmallmatrix}
    \endgroup
=(11)
\end{equation}

The set of linear combinations can be aggregated in the $5\times5$ matrix $M=(M_{17}|M_{11}|M_{29})$ allowing to obtain the global relationship between the path label $P$ and the concatenation of interface labels list $L$:
\begin{equation}
    P\cdot M = L
\end{equation}
By observing that $M$ is invertible, the source computes $M^{-1}$ and $P$ as follows:
\begin{equation}
    P = L\cdot M^{-1} = (10\ 1\ 11)\cdot 
        \begingroup 
    \setlength\arraycolsep{2pt}
    \begin{psmallmatrix} 0 & 0 & 1 & 1 & 1\\ 1 & 0 & 0 & 1 & 1\\ 1 & 1 & 1 & 1 & 1\\ 1 & 1 & 0 & 0 & 1\\ 1 & 1 & 1 & 0 & 1 
    \end{psmallmatrix}
    \endgroup
= (11100)
\end{equation}

It can be verified that (\ref{eq:filtering1}), (\ref{eq:filtering2}) and (\ref{eq:filtering3}) hold for this value of $P$. 

\section{XOR-based Source Routing}
\label{sec:description}

We have previously illustrated within a little example the main principle of our proposal. We now present in further details the core mechanisms of XOR-based source routing.

\subsection{Network Hypotheses}
\label{ssec:hypotheses}

We define a network has a set of edge nodes (source and/or destination nodes) connected to $n_R$ routers $R_j,\ j=1,\ldots,n_R$ as illustrated Fig.~\ref{fig:ExNetwork}. Communications occur between several edge nodes through a path formed by several routers. For example, a unicast communication between a source $S$ and destination $D$ could use either the path  $(R_{17}|R_{50}|R_{29})$ or $ (R_{17}|R_{12}|R_{39}|R_{29})$. The connection can be unicast, multipath or multicast.

\subsection{Router Forwarding}
\label{ssec:router}

The main operation done by a router $R_i$ at the reception of a packet is to filter the path label $P$ to recover its corresponding interface label $L_i$. 

The general form of the simple example presented in \ref{sec:XSRnutshell} is to consider that each router stores $2^\epsilon$ binary filtering matrices $\overline{M}_{i}^{(e)}$, for $e=0,\ldots,2^\epsilon-1$. Each matrix has $\overline{s_P}$ rows and $\overline{s_i}$ columns, where $\overline{s_P}$ is the maximum length in bits of the size of $P$ and $\overline{s_i}$ is the maximum size of the interface labels of $R_i$.

At the reception of a packet, the router reads the path label $P$ of size $s_P$ and the value $e$ set by the source stored on $\epsilon$ bits (further explained in the following). 

Since $s_P\le \overline{s_P}$ and $s_i\le \overline{s_i}$, the router takes $M_i$ as the submatrix of $\overline{M}_{i}^{(e)}$ formed by the first $s_l$ rows and $s_i$ columns.   

The filtering function $F_i$ is a set of linear operations which consists in multiplying the path label $P$ by a filtering matrix $M_i$ with $s_i$ columns and $s_P$ rows. 

More formally, if $\mathbb{F}_2$ denotes the binary finite field, the filtering function is defined as follows:
\begin{displaymath}
F_i \left\{ \begin{array}{ccc}
\mathbb{F}_2^{s_P} & \longrightarrow & \mathbb{F}_2^{s_i}\\
P & \longrightarrow & P\cdot  M_i\\
\end{array} \right.
\end{displaymath}

These operations are summarized in Fig.~\ref{fig:receiverProcess}
\begin{figure}[htb]
	\begin{center}
		\includegraphics[width=0.45\textwidth]{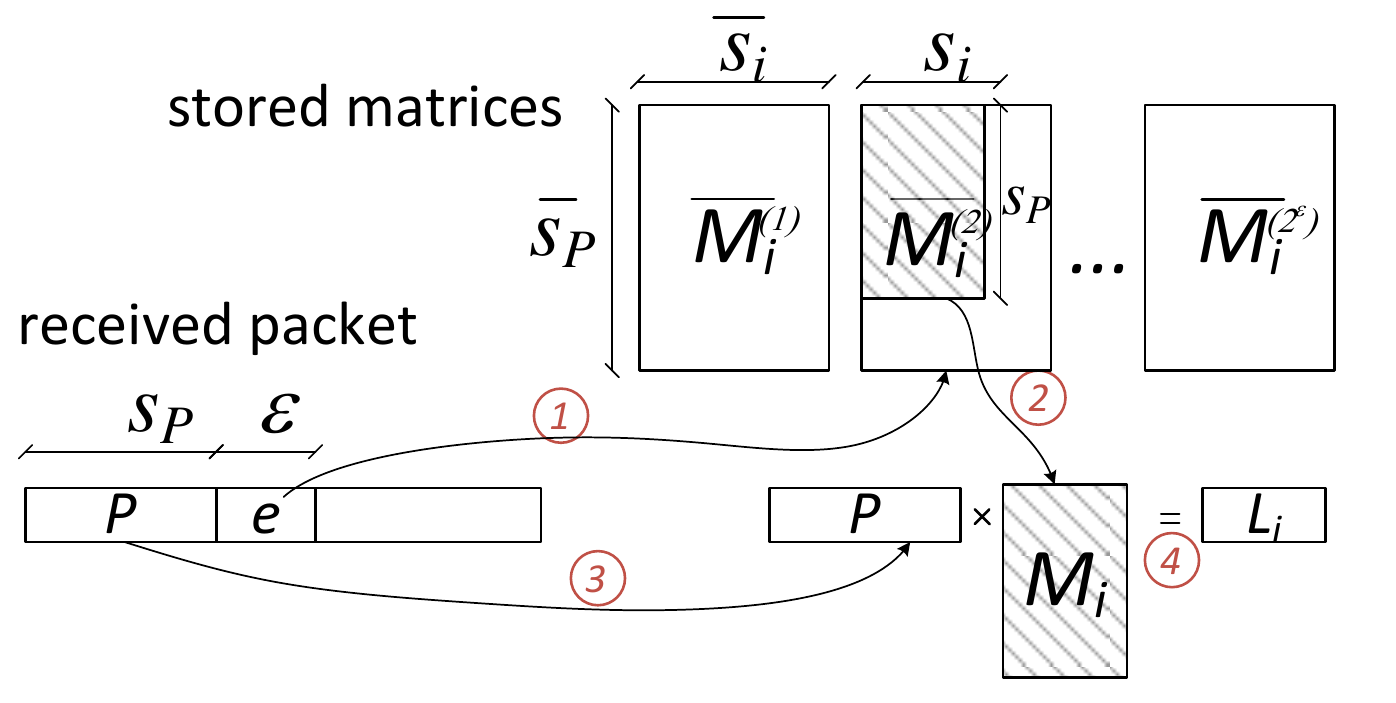}
	\end{center}
	\caption{Receiver operations. The numbers represent the different steps.}
	\label{fig:receiverProcess}
\end{figure}

\subsection{Path Label Construction}
\label{ssec:pathLabelConstruction}

Once again, the construction of a path label from a source $S$ to a destination $D$ is done either by the source itself or by the controller which builds the path label and send it to the source. 

Let $\{R_{i_1},\ R_{i_2},\ldots, R_{i_p}\}$ the set of the routers on the path. For unicast transmission, this set corresponds to a sequence of routers. Considering multicast, this sequence is not ordered and corresponds to the set of routers that will forward the packet.

The concatenation of the interface labels corresponds to a bit vector denoted $L=(L_{i_1},\ L_{i_2},\ldots, L_{i_p})$ with a size $s_L=\sum_{k=1}^p s_{i_k}$. 

As seen in the previous paragraph, the routers multiply the path label by their filtering matrix to obtain their interface label. This implies that the path label $P$ must verify some linear equations. These equations can be represented by the matrix $M=(M_{i_1},\ M_{i_2},\ldots, M_{i_p})$ which is the concatenation of the filtering matrices used by the routers of the path. Since we set $s_P$ the length of $P$, to $s_L$, $M$ is a $s_L\times s_L$-square matrix. 
We define a path label as \emph{valid} if the filtering process (defined in the previous section) applied by any router of the path produces the correct interface label of a given router. We will show that we can obtain a valid path label $P$ if $M$ is nonsingular.

The construction of $P$ is based on the following theorem:
\begin{theorem}
	\label{thm:nonsingular}
	Let $M^{-1}$ be the inverse of M. Then:
	\begin{equation}
		P \overset{\underset{\mathrm{def}}{}}{=} L\cdot M^{-1}
	\end{equation}
	is a valid path label. 
\end{theorem}
\begin{proof}
	From the definition of $P$, we have $P\cdot M=L\cdot M^{-1}\cdot M=L=(L_{i_1},\ L_{i_2},\ldots, L_{i_p})$. On the other hand, $P\cdot M=P\cdot (M_{i_1},\ M_{i_2},\ldots, M_{i_p})=(P\cdot M_{i_1},\ P\cdot M_{i_2},\ldots, P\cdot M_{i_p})$. It follows that, for each $k=1,\ldots,p$, $P\cdot M_{i_k}=L_{i_k}$ and thus $P$ is valid. 
\end{proof}

\begin{theorem}
	\label{thm:returnPath}
	Let $P$ be a valid path built from Theorem \ref{thm:nonsingular} to route the packets for a unicast transmission from a sender to a destination. 
	
	Then, $P$ is also valid to route the packets from the destination to the source. 
\end{theorem}
\begin{proof}
    To prove this theorem, it is sufficient to prove that if a router receives a packet with a path label $P$ on an input interface $ID_i$ and forwards it to the output interface $ID_o$, then, if it receives a packet with the same path label on the interface $ID_o$, it forwards it to the interface $ID_i$. 

     Let us consider a router $R_u$ of the path. On the forward path, the router filters the path label of a packet with the matrix $M_u$ and obtains the interface label $L_u=P.M_u$. According to \ref{sec:XSRnutshell}, $L_u$  is the XOR of the input interface $ID_i$ and the output interface $ID_o$. Since it knows $ID_i$, it is able to recover the $ID_o$ by XORing $L_u$ and $ID_i$ ($ID_i\oplus L_u= ID_i\oplus (ID_i\oplus ID_o)= ID_o$). Then it transmits the packet on the output interface. 
     
     On the reverse path, it receives a packet from the interface $ID_o$ with the same path label. By applying the filtering function $M_u$, it recovers $L_u$. Since it knows the packet arrived from interface $ID_o$, it XORs $ID_o$ and $L_u$ and obtains $ID_i$ ($ID_o\oplus L_u= ID_o\oplus (ID_i\oplus ID_o)= ID_i$. Then it forwards the packet to the interface $ID_i$. 
\end{proof}

\section{Analysis of the Parameters}
\label{sec:analysis}

The main system parameter that must be evaluated is the probability of building a valid path label. Indeed, this probability impacts on the complexity of the construction of a path label and allows to correctly size $\epsilon$ previously presented in \ref{ssec:router}.
Note we make no assumption on the filtering matrices and consider them as random binary matrices.

\subsection{Probability of Construction of a Valid Path Label}
\label{ssec:proSuccValid}

According to \ref{ssec:pathLabelConstruction}, a valid path label can be built if the matrix $M$ is invertible.  \cite{berlekamp:80} shows that the probability that a $s_L\times s_L$ binary random matrix is invertible is equal to 
\begin{equation}
    \sigma_0(s_L) = \sum_{i=0}^{s_L-1} (1-2^{s_L-i})
\end{equation}
Fig. \ref{fig:invertMat} confirms that this probability quickly converges to the limit which is known to be $0.2888$.

A valid path can be build if at least one of the $2^\epsilon$ matrices built from the matrices stored by the routers is invertible. The probability is thus equal to \begin{equation}
    \sigma_\epsilon(s_L) = 1 - (1-\sigma_0(s_L))^{2^\epsilon}
\end{equation}
These values are plotted Fig. \ref{fig:validPath}. It can be observed that a valid path can be obtained with a very high probability (for example, $0.99998$ for $2^5=32$ $ 8\times 8$ binary matrices stored by each router). 
\begin{figure}[!htb]
    \centering

    \begin{minipage}{0.24\textwidth}
        \centering
        \includegraphics[width=\linewidth]{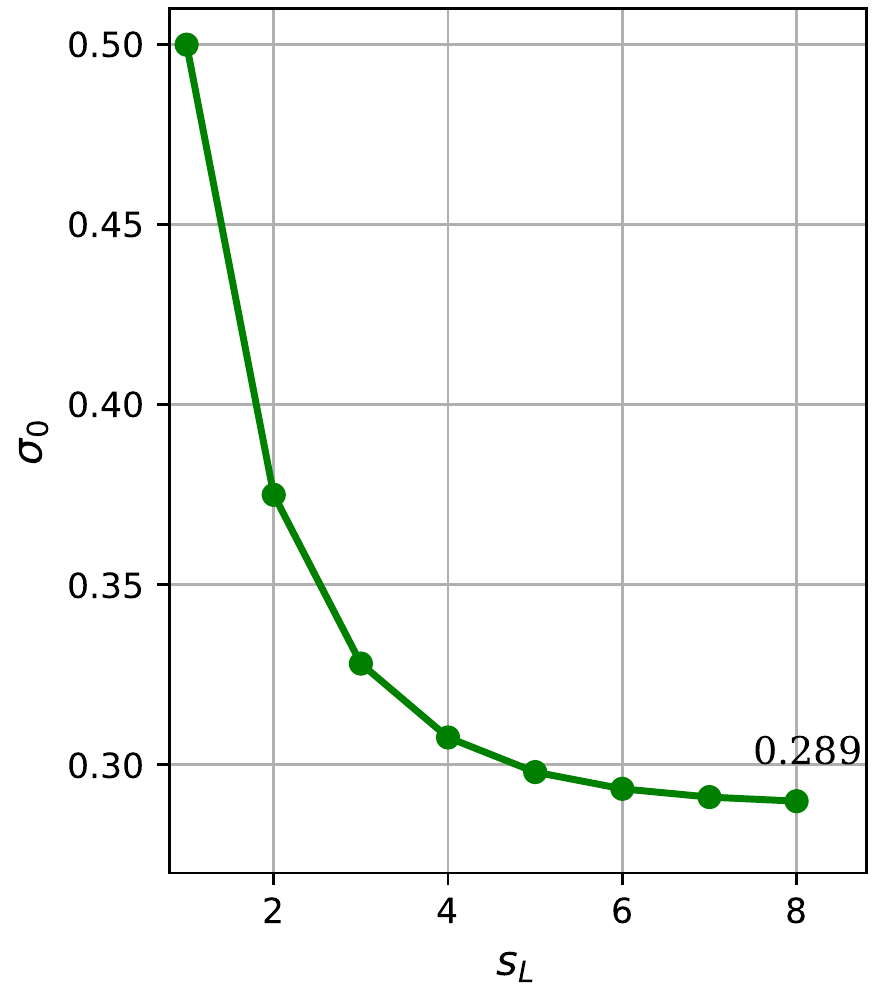}
        \caption{proba. invertible matrix}
        \label{fig:invertMat}
    \end{minipage}
    \begin{minipage}{0.24\textwidth}
        \centering
        \includegraphics[width=\linewidth]{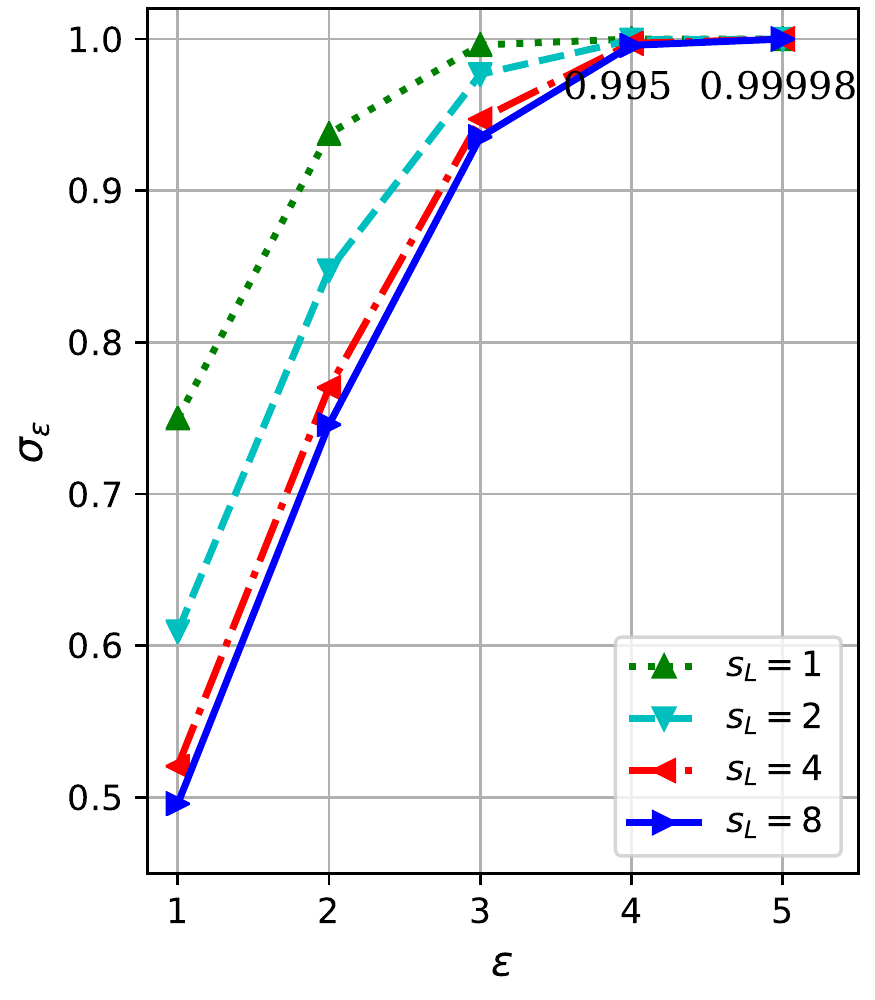}
        \caption{proba. valid path}
        \label{fig:validPath}
    \end{minipage}
\end{figure}

\subsection{Complexity of the Path Label Construction}
\label{ssec:complPathLabelConstruction}

As the path label computation can also be done in the control plane (e.g. SDN controller), there is no impact on the data plane forwarding procedure. However, we believe that estimating its complexity is of interest.

To build a path label from a matrix $M$, it is necessary to find its rank and to perform a matrix inversion. This is generally done by using algorithms based on Gaussian Elimination (GE). Even if there exists theoretical optimizations like  Strassen's algorithm \cite{Strassen:1969} which runs in $\mathcal{O}(n^{2.807})$ operations for very large matrices, we will consider that the complexity is in  $\mathcal{O}(n^3)$ for each tested matrix.

To evaluate the average number of tested matrices, we can observe that it is necessary to perform $k$ GEs only when the first $k-1$ fail to build a valid path and when the $k^{th}$ succeeds. This occurs with the probability $(1-\sigma)^{k-1}\sigma$. It follows that the average number of GEs is: 
\begin{equation}
\sum_{k=1}^{\infty} k(1-\sigma)^{k-1}\sigma = \frac{1}{\sigma}
\end{equation}
A value of $\sigma=0.6103$ gives an average number of GEs equal to $1/0.6103=1.6385$.

\subsection{Number of Signalling Bits}
\label{ssec:packetHeaderSize}
Section \ref{ssec:proSuccValid} has shown that $\sigma$ must be chosen greater or equal to $3$ to provide a high probability of building a path label. This represents the size of the signalling field added to the packet header with the path label. For example, this value is similar to the size of the pointer used in \cite{2017:optimalPathEncoding} which is $4$ bits in most of studied configurations. 

\subsection{Storage Amount in Routers}
\label{ssec:storageRouters}
According to \ref{ssec:router}, each router stores $2^\epsilon$ binary matrices of $\overline{s_P}$ rows and $\overline{s_i}$ columns. So the global amount stored by a router is $$2^\epsilon\times \overline{s_P} \times \overline{s_i}$$ 

By considering unicast transmissions, reasonable maximal values of $\epsilon=4$, $\overline{s_P}=50$ and $ \overline{s_i}=10$ can be chosen. This represents $50\times 10\times 16=8000$ bits, \textit{i.e} $1000$ bytes which is completely scalable.

For multicast transmissions, the path label can be larger (see \ref{ssec:datacenter}). In the largest studied case, the path label has a size of around $200$ bytes and the interface labels have a maximal size of around $100$ bits. For $\epsilon=4$, the total number of bits stored is $16\times 200\times 8\times 100= 2.56$ Mbits \textit{i.e} $320$ Kbytes. Even if this number is rather low compared to traditional routers, we can easily reduce it by globally optimizing the choice of the filtering matrices in order to reduce the value of $\epsilon$. Another possibility is to use filtering matrices that can be deduced from a short representation as in \ref{sec:XSRnutshell} where the columns of the filtering matrices are deduced from the first column by cyclic permutations.

\section{XSR Versus Existing Work}
\label{sec:comparisons}

Two recent results have interesting relationships with our proposal. In the two next sections, we expose these links and compare various metrics of interest.

\subsection{Optimal Path Encoding for Unicast Transmissions}
\label{ssec:OPE}
This first considered work, denoted OPE, is presented in \cite{2017:optimalPathEncoding}. The authors propose to use prefix codes to represent the interface labels and optimize the choice of the labels in order to minimize the maximal length of the path label. The path label is then the sequence of the interface labels with an additional pointer indicating to a router the position in the encoded path that it must consider. This pointer is updated by router according to the length of its interface label. This scheme allows to reduce significantly the size of the largest path label. 

Compared to this work, our proposal goes further by encoding their output (the sequence of optimized interface labels) with binary linear operations. 

If we estimate the amount of bits needed to implement each solution, the lengths of the path labels are equal both for OPE and XSR and require the same amount of signalling bits: around $4$ for OPE to encode the pointer and $\epsilon =3$ or $4$ with XSR.     

However, the advantage to add XSR on top of OPE is twofold:
\begin{enumerate}
    \item the pointer used by OPE involves ordered sequence of interface labels and thus can only be used for unicast transmissions. This is rather unfortunate because the idea of optimizing the interface labels according to the maximal length makes sense for multicast transmissions as in datacenter networks (see next section). Encoding the path with XSR removes this notion of order and thus allows multicast transmissions;   
    \item using fast filtering router operations allows to prevent any packet modification due to pointer update or possibly integrity checks.  
\end{enumerate}

\subsection{Datacenter Networks}
\label{ssec:datacenter}

\subsubsection{Recent Work in Source Routing for Datacenter Networks}
\label{sssec:IntrodataCenter}    

the potential of source routing for datacenter networking was demonstrated in KeyFlow \cite{keyflow} and COXcast \cite{coxCast2014} for both unicast and multicast transmissions. The first interest is the simplification of the management of multiple small multicast groups. The protocol Xcast \cite{rfc7855} was defined for this purpose. However, the generated headers can be large. To cope with this issue, KeyFlow and COXcast independently propose a source routing mechanism encoding the paths with interface labels associated to the interfaces of the routers. The main idea is to associate to each router a prime number label and to the paths an integer stored in the packet header. At the reception of a packet, a router simply computes the residue of the path modulo its label. The obtained value corresponds to the output interface(s). They reduce significantly the size of the path label compared to Xcast. Moreover, the core routers neither use forwarding tables nor modify the packets. This simplifies router operations and reduces the processing delay allowing ultra-low latency communications. 

The path label size is also reduced in the RDNA architecture \cite{rdna2018}. RDNA improves the way to choose the prime numbers and to compute the path. Since the integer path is determined from the prime numbers of the system, it is preferable to use short prime numbers in order to reduce the size of the integer path. Unfortunately, the amount of primes in integer numbers is quite low and it is not always possible to choose small prime numbers that provide residues with a given number of bits. Multiplying prime numbers (and finding the right ones) leads to oversized binary values, making the path label size not optimal and RDNA solution less flexible than XSR in particular in the context of multicast.

\subsubsection{Path Length Comparison}
\label{sssec:dataCenterComparison}    

the mechanism used in COXcast and RDNA lays on a concept similar to XSR. The main difference is that XSR is based on linear algebra while both others are based on modular arithmetic.
Linear algebra leads to several advantages:
\begin{itemize}
	\item linear algebra does not have the problem of scarcity of prime numbers and thus the length of the path is very close to the optimal. This is demonstrated in Tables \ref{tab:unicast} and \ref{tab:multicast};
	\item linear operations performed in routers are simple dot products and are less complex than modulo operations on integers;
	\item linear algebra provides a better flexibility. The configuration of the global network can easily be changed because finding new invertible matrices is effortless and leads to optimal size compared to the complex choice of the best set of prime numbers. 
\end{itemize}   

We now compare the overhead in terms of sizes. We consider the use cases studied in RDNA \cite{rdna2018} and compute the corresponding header length for each solution. 

The datacenter network analyzed in \cite{rdna2018} is a 2-tier Clos network topology (shown Fig.~\ref{fig:rdnanetwork}) composed of two stages core switches (spine and leaf) and one stage of edge switches connected to hosts. The connections are defined between two hosts. The considered path is defined between the edge switches respectively connected to the hosts source and  destination. The longest path is from the edge switch connected to source to the edge switch connected to the destination through a first leaf, a spine and a second leaf.   

\begin{figure}[htb]
	\begin{center}
		\includegraphics[width=0.35\textwidth]{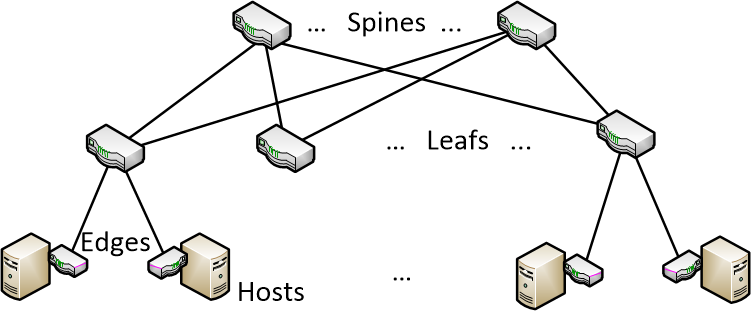}
	\end{center}
	\caption{Datacenter network}
	\label{fig:rdnanetwork}
\end{figure}

Let us denote $\texttt{spines}$, resp. $\texttt{leafs}$, the number of spines, resp. leafs,  and let $\texttt{ports}$ be the number of ports of the leafs. The number of ports of the spines is $\texttt{leafs}$. 

To represent a unicast path, we then need to store the output port of the first leaf (i.e. among $\texttt{ports}-1$ since we do not consider the input port), then the output port of the spine (among $\texttt{leafs}-1$) and then the output port of the second leaf (among $\texttt{ports}-1$). 

According to the results of \ref{ssec:pathLabelConstruction}, this path can be encoded in $2.\log_2(\texttt{ports}-1)+\log_2(\texttt{leafs}-1)+\epsilon$ bits with a high probability. We fix the value of $\epsilon$ to $4$ bits\footnote{The value of $\epsilon$ can be reduced by determining a static configurations of the filters (out of the scope of this paper).}.

The number of bytes necessary to encode the path is thus:
\begin{equation}
\big\lceil {(2.\log_2(\texttt{ports}-1)+\log_2(\texttt{ports}-1)+4)}/{8}\big\rceil \nonumber
\end{equation} 

The obtained values are compared to COXcast and RDNA in Table \ref{tab:unicast}. We observe that we always have path label sizes lower or equal to the other proposals.

\begin{table*}[t]
	\caption{\label{tab:unicast} Path label size (bytes) for unicast }
	\centering
	\begin{small}
	\begin{tabular}{ | c || c | c | c | c | c | c | c | c | c | c | c | c | c | c | c | c | c | c | c | c | c | c | }
		\hline
		Spine & \multicolumn{3}{c|}{2} & \multicolumn{5}{c|}{6} & \multicolumn{5}{c|}{12} &  \multicolumn{5}{c|}{8}   \\ \hline 
		Leafs &  \multicolumn{3}{c|}{4} & \multicolumn{5}{c|}{12} & \multicolumn{5}{c|}{16} & \multicolumn{5}{c|}{16}    \\ \hline  
		Ports & 16 & 24 & 32 & 16 & 24 & 32 & 48 & 96 &  16 & 24 & 32 & 48 & 96 &  16 & 24 & 32 & 48 & 96   \\ \hline \hline 
		COXcast\cite{coxCast2014} & 5 & 8 & 11 & 5 & 8 & 11 & 17 & 35 &  5 & 8 & 11 & 17 & 35 &5 & 8 & 11 & 17 & 35   \\ \hline 
		RDNA \cite{rdna2018} & 2 & 2 & 2 & 3 & 3 & 3 & 4 & 4 &  3 & 3 & 3 & 4 & 4 & 3 & 3 & 3 & 4 & 4   \\ \hline 
		XSR  & 2 & 2 & 2 & 2 & 3 & 3 & 3 & 3 & 2 & 3 & 3 & 3 & 3 & 2 & 3 & 3 & 3 & 3   \\ \hline
	\end{tabular}
\end{small}
\end{table*}

For multicast transmissions, the longest path is from the host source and its corresponding edge switch to all other hosts. The packet must be sent from the corresponding edge switch to a first leaf which forwards it to all its ports connected to other edge switches and to one spine. The spine transmits the packets to all others leafs which forwards it to all their connected edge switches (see Fig.~4 of \cite{rdna2018}).

We recall that multicast interface labels can be represented as a bitmap of the output ports. Thus, a multicast interface label of a spine is a vector of $\texttt{leafs}$ bits and an interface label of a leaf is a vector of $\texttt{ports}$ bits.

The application of results of \ref{ssec:pathLabelConstruction} leads to a encoded path of length \texttt{ports}+\texttt{leafs}+(\texttt{leafs}-1).\texttt{ports}+$\epsilon$ bits. By fixing the value of $\epsilon$ to $4$ bits, we obtain the following number of bytes:
\begin{equation}
    \big\lceil {(\texttt{ports}+\texttt{leafs}+(\texttt{leafs}-1).\texttt{ports}+4)}/{8}\big\rceil \nonumber
\end{equation} 

The values obtained are reported in Table \ref{tab:multicast} in the row "XSR v1". Except two cases, a small gain is observed in most configurations.  

\begin{table*}[t]
	\caption{\label{tab:multicast} Path label size (bytes) for multicast }
	\centering
	\begin{small}
		\begin{tabular}{ | c || c | c | c | c | c | c | c | c | c | c | c | c | c | c | c | c | c | c | c | c | c | c | }
			\hline
			Spine & \multicolumn{3}{c|}{2} & \multicolumn{5}{c|}{6} & \multicolumn{5}{c|}{6} &  \multicolumn{5}{c|}{8}   \\ \hline 
			Leafs &  \multicolumn{3}{c|}{4} & \multicolumn{5}{c|}{12} & \multicolumn{5}{c|}{16} & \multicolumn{5}{c|}{16}    \\ \hline  
			Ports & 16 & 24 & 32 & 16 & 24 & 32 & 48 & 96 &  16 & 24 & 32 & 48 & 96 &  16 & 24 & 32 & 48 & 96   \\ \hline \hline 
			COXcast\cite{coxCast2014} & 10 & 14 & 18 & 36 & 48 & 60 & 84 & 156 & 47 & 63 & 79 & 111 & 207 & 51 & 67 & 83 & 115 & 211   \\ \hline 
			RDNA \cite{rdna2018} & 9 & 14 & 18 & 26 & 39 & 52 & 75 & 154 &  34 & 51 & 68 & 100 & 200 & 34 & 51 & 68 & 100 & 200   \\ \hline 
			XSR v1 & 9 & 13 & 17 & 26 & 38 & 50 & 74 & 146 & 35 & 51 & 67 & 99 & 195 & 35 & 51 & 67 & 99 & 195  \\ \hline
			XSR v2 & 9 & 13 & 17 & 20 & 32 & 44 & 68 & 140 & 26 & 42 & 58 & 90 & 186 & 22 & 38 & 54 & 86 & 182  \\ \hline
		\end{tabular}
	\end{small}
\end{table*}

The rather intuitive representation of our filtering operation allows us to propose an enhancement of the path encoding in the multicast case. The idea is to optimize the interface label in the leafs by differentiating  the ports of the leafs connected to the spines and the one connected to edge switches. We propose to use some "signalling" bits in the label to encode differently the packets that must be only sent to some spines, the ones that must be only sent to edge switches and the others. To reduce the number of these bits, we use the prefix code $\{0,10,11\}$ as it was proposed in \cite{2017:optimalPathEncoding} for unicast transmissions.  

In a use case, for $\texttt{spines}=6$ and $\texttt{ports}=16$, the new labels would be:
\begin{itemize}
	\item $[10......]$: the code $10$ followed by the $6$ ports connected to the spines for the packets only sent to some spines
	\item $[0..........]$: the code $0$ followed by the $16-6=10$ ports connected to the edges  for the packets that only be sent to the edges.
	\item $[11................]$: the code $11$ followed by the $16$ ports for the others packets
\end{itemize}

This leads to an encoded path of length $2+\texttt{ports}+\texttt{leafs}+(\texttt{leafs}-1).(1+\texttt{ports}-\texttt{spines})+\epsilon$ bits. By fixing the value of $\epsilon$ to $4$, we obtain the following number of bytes:
\begin{eqnarray}
    \big\lceil (2+\texttt{ports}+\texttt{leafs}+(\texttt{leafs}-1).\nonumber \\
    (1+\texttt{ports}-\texttt{spines})+4)/{8}\big\rceil \nonumber
\end{eqnarray} 

From a practical point of view, to forward a packet with the types of labels, the leaf just needs to identify the $2$ first bits to determine the length of the label it must recover and the filter it must use.

The results reported in the row "XSR v2" of Table \ref{tab:multicast} show a significant gain in most use cases.     

\section{Conclusion and Future Work}
\label{sec:conclusion}

We presented XOR-based source routing, a new data plane scheme enabling fast forwarding by performing only simple linear operations over a binary vector label which embeds an encoded routing path label. Compared to recent approaches, XSR computes the smallest label possible and does not need to modify forwarded packets. The main advantage compared to other existing approaches is to allow the re-use of the same path label for the feedback path and so, prevent the receiver to compute another label to reply (considering the SDN controller allows the same path for reply). XSR provides the building blocks to speed up the forwarding plane and can be applied to different data planes such as MPLS or IPv6 for unicast and multicast communications.

In a future work, we expect to implement XSR within Mininet emulator to further demonstrate the effective processing cost of forwarding operations.
Furthermore, we believe that XSR would lead to promising application in terms of privacy and security if routers filtering operations remain unknown to attackers attempting to observe the network. Last but not least, we wish to present and discuss this solution at the IETF.

\bibliographystyle{plain}
\bibliography{bibRouting}

\end{document}